\documentclass[11pt,tgrind,psbox]{article}
\usepackage{latexsym}
\usepackage{graphics}
\usepackage{amsmath}
\usepackage{xspace}
\usepackage{amssymb}
\usepackage{psfrag}
\usepackage{epsfig}
\usepackage{fullpage}
\usepackage{amsmath,amsthm}
\usepackage{epsfig}
\usepackage{pst-all}

\newcommand {\bel}[1]{\begin{align*}}
\newcommand {\eel}[1]{\end{align*}}
\newcommand {\bea}{\begin{eqnarray}}
\newcommand {\eea}{\end{eqnarray}}

\def\myfrac#1#2{
    \hspace{3pt}\!\!\!$_{#1}\!\!\hspace{1pt}\backslash\hspace{2pt}\!\!^{#2}$\!\!\!\hspace{3pt}
    }

\newcommand{\pr}{\mathbb{P}}
\newcommand{\E}{\mathbb{E}}

\newcommand{\ignore}[1]{\relax}

\def\rZ{{\mathbb Z}}

\newtheorem{theorem}{Theorem}
\newcommand{\remark}{{\bf Remark : }}
\newtheorem{lemma}{Lemma}

\newtheorem{proposition}{Proposition}
\newtheorem{corollary}{Corollary}

\newtheorem{definition}{Definition}

\begin{document}
\title{Randomized greedy algorithms for independent sets and matchings in regular
graphs: Exact results and finite girth corrections.}
\author{
{\sf David Gamarnik }
\thanks{Operations Research Center and Sloan School of Management, MIT, Cambridge, MA,  02139, e-mail: {\tt
gamarnik@mit.edu}}
\and
{\sf David A. Goldberg}
\thanks{Operations Research Center, MIT, Cambridge, MA,
02139, e-mail: {\tt dag3141@mit.edu}}
}
\maketitle
\begin{abstract}
We derive new results for the performance of a simple greedy algorithm for finding large
independent sets and matchings in constant degree regular graphs.  We
show that for $r$-regular graphs with $n$ nodes and girth at least $g$, the algorithm
finds an independent set of expected cardinality $f(r)n -
O\big(\frac{(r-1)^{\frac{g}{2}}}{ \frac{g}{2}!} n\big)$, where $f(r)$ is a function which
we explicitly compute. A similar result is established for matchings.
Our results imply improved bounds for the size of the largest independent set in these graphs, and
provide the first results of this type for matchings. As an implication we show that the greedy algorithm
returns a nearly perfect matching when both the degree $r$ and girth $g$ are large.
Furthermore, we show that the cardinality of independent sets and matchings produced by the greedy algorithm
in \emph{arbitrary} bounded degree graphs is concentrated around the mean.
Finally, we analyze the performance of the greedy algorithm for the case of random i.i.d. weighted independent
sets and matchings,
and obtain a remarkably simple expression for the limiting expected values produced by the algorithm.
In fact, all the other results are obtained as straightforward corollaries from the results for the weighted case.
\end{abstract}
\pagenumbering{arabic}

\section{Introduction}\label{section:introsec}
\subsection{Regular graphs, independent sets, matchings and randomized greedy algorithms}
 An $r$-regular graph is
a graph in which every node has degree exactly $r$.  The girth $g$
of a graph is the size of the smallest cycle.  Let $G(g,r)$ denote
the family of all $r$-regular graphs with girth at
least $g$.  For a graph $G$, we denote the set of nodes  and edges  by
$V(G)$ and $E(G)$, respectively.  A set of nodes $I$ is defined to be
an independent set if no two nodes of $I$ are adjacent. For a
graph $G$, let ${\cal I}(G)$ denote (any) maximum cardinality
independent set  ($MIS$) of $G$, and $|{\cal I}(G)|$ its cardinality.
Throughout the paper we will drop the explicit reference to the underlying
graph $G$ when there is no ambiguity. For example we use ${\cal I}$ instead of
${\cal I}(G)$ or $V$ instead of $V(G)$.

Suppose the nodes of a graph are equipped with some non-negative
weights $W_i, 1\leq i \leq n\triangleq |V|$.  The weight $W[I]$ of a given independent set $I$ is
the sum of the weights of the nodes in $I$.  When the nodes of a graph are equipped with weights which
are generated i.i.d. using a continuous distribution function $F(t) = \pr(W_i \leq t)$ with
non-negative support, we denote by ${\cal I}_W$ the random unique with probability $1$ (w.p.1)
maximum weight independent set $(MWIS)$ of $G$.

A (partial) matching is a set of edges $M$ in a graph $G$ such that every node
is incident to at most one edge in $M$.
For a graph $G$, let ${\cal M}$ denote (any) maximum cardinality matching ($MM$) of $G$.
Suppose the edges of a graph are equipped with some non-negative weights $W_e, e\in E$.
The weight $W[M]$ of a given matching $M$ is
the sum of the weights of the edges in $M$.
When the edges of a graph $G$ are equipped with weights generated i.i.d. using a continuous distribution function $F$ with
non-negative support, we denote by ${\cal M}_W$ the random unique (w.p.1) maximum weight matching $(MWM)$
of $G$.

In this paper we analyze the performance of a simple greedy
algorithm, which we call $GREEDY$,  for finding large independent sets and matchings.
The description of the $GREEDY$ algorithm is as follows. For independent sets, $GREEDY$
iteratively selects a node $i$ uniformly at random (u.a.r) from all
remaining nodes of the graph, adds $i$ to the independent set,
deletes all remaining nodes adjacent to $i$ and repeats. Note that while the underlying
graph is non-random, the independent set produced by $GREEDY$ is random as it is based on
randomized choices. For
$MWIS$, $GREEDY$ iteratively selects the node $i$ with the  greatest
weight from all the remaining nodes, adds $i$ to the independent set,
deletes all the remaining nodes adjacent to $i$ and repeats.  Note that when acting on a fixed weighted graph, the action of $GREEDY$ is non-random.  In this setting, the randomness will come from the fact that the weighting itself is i.i.d.  For matchings $GREEDY$
operates similarly, except that it chooses edges instead of nodes,
and deletes edges incident to the chosen edge.

Let ${\cal IG}({\cal MG})$ denote the random independent set (matching) returned by
$GREEDY$ when run on an unweighted or (randomly) weighted graph $G$, depending on context.
Denote by $W[{\cal IG}] (W[{\cal MG}])$
the weight of ${\cal IG} ({\cal MG})$ (for the weighted case), and by $|{\cal IG}|(|{\cal MG}|)$ the respective cardinalities (in the unweighted case).
Our goal is obtaining bounds on
the expectation and variance of
$|{\cal IG}|$, $|{\cal MG}|$, $W[{\cal IG}],$ $W[{\cal MG}]$, where the latter two will be considered
for the case of i.i.d. continuous non-negative weight distributions. One of the
motivations is to derive new lower bounds on largest
 independent set in
constant degree regular graphs with large  girth.

\subsection{Summary of our results and prior work}
Our main results are Theorems~\ref{theorem:MWISWEIGHT},\ref{theorem:MWMWEIGHT} which provide
remarkably explicit upper and lower bounds on the expected weight of the independent set and matching produced by $GREEDY$
in a regular graph of large fixed girth when the weights are generated i.i.d. from a continuous non-negative distribution. Since the gap between
the upper and lower bound is of the order $\approx (r-1)^{g/2}/(g/2)!$, we also obtain the limiting expression
for the weight of the independent set and matching produced by $GREEDY$ in a regular graph when
the girth diverges to infinity. These results are Corollaries~\ref{coro:MWISWEIGHT},\ref{coro:MWMWEIGHT}.

As a corollary we obtain upper and lower bounds on $\E[|{\cal IG}|]$ and $\E[|{\cal MG}|]$, by considering a uniform distribution which
is highly concentrated around $1$. These results are stated as Theorems~\ref{theorem:MIS},\ref{theorem:MM}.
Again the gap between the upper and lower bounds is of the order $\approx (r-1)^{g/2}/(g/2)!$ and we obtain
a limiting expression when the girth diverges to infinity, as stated in Corollaries~\ref{coro:MICORunweight},~\ref{coro:MMCORunweight}.
While Corollary~\ref{coro:MMCORunweight} is a new result, Corollary~\ref{coro:MICORunweight} is not. This result was recently
established by Lauer and Wormald ~\cite{L.} using a different approach called the `nibble'  method.
Thus our Theorem~\ref{theorem:MIS} can be viewed as an explicit finite girth correction
to the limiting result (Corollary ~\ref{coro:MICORunweight}) derived earlier in~\cite{ L.} and proved here using different methods.

Our results on the performance of the $GREEDY$ algorithm, as well as
the results of \cite{ L.}, are motivated by the problem of obtaining
lower bounds on the size of the largest independent set in regular
graphs, and specifically regular graphs with large girth. The
history of this problem is very
long~\cite{HopkinsStaton},\cite{S.83},\cite{S.91},\cite{S.95} with
\cite{ L.} being the latest on the subject. In particular, the lower
bounds obtained in \cite{ L.} are the best known for the case $r\ge
7$ and sufficiently large girth, and in this range they beat
previous best bounds obtained by Shearer~\cite{S.91}.   Although
these bounds are the best known as the girth diverges to infinity
(for any fixed $r$), the bounds given in \cite{L.} for any fixed
girth are very difficult to evaluate, as they are given implicitly
as the solution to a large-scale optimization problem.  Our bounds
match those of \cite{L.} for any fixed $r$ as the girth diverges to
infinity, \emph{and} give simple explicit bounds for any fixed girth
as a finite girth correction of the order  $\approx
(r-1)^{g/2}/(g/2)!$.  In addition, our bounds are superior for
several instances discussed in \cite{L.} where bounds were derived
numerically by lower-bounding the aforementioned large-scale
optimization problem. The details of this comparison are presented
in Section~\ref{section:NUMERICS}.

Our corresponding results for matchings  are new, both the limiting
version, Corollary~\ref{coro:MMCORunweight}, and the finite girth
correction, Theorem~\ref{theorem:MM}.  Interestingly, by considering
the upper and lower bounds in Theorem~\ref{theorem:MM} and taking a
double limit $r,g\rightarrow\infty$, we find that the $GREEDY$
algorithm produces a nearly perfect matching in the double limit
$r,g\rightarrow\infty$. This partially answers an open problem posed
by Frieze~\cite{FG} regarding the construction of a simple,
decentralized algorithm for finding a nearly perfect matching in
constant degree regular graphs with large girth.

Our second set of results, Theorems~\ref{theorem:VAR12} and~\ref{theorem:VAR11},
concerns the variance of the weight (cardinality) of the independent
set and matching produced by $GREEDY$ in \emph{arbitrary} graphs
with bounded degree. That is no additional assumptions on girth or
regularity are adopted. We show that when the weights are i.i.d. and have finite
second moment, and when the graph has bounded degree, the variance,
appropriately normalized, is of the order $O(1/n)$ in both cases.
We are also able to
give explicit bounds in terms of the graph degree, the number of
nodes, and the second moment of the weighting distribution.  We also
give similar results for the unweighted case.  Thus the answers
produced by $GREEDY$ are highly concentrated around their means, and in
this sense the $GREEDY$ algorithm is very robust.  We believe these
are the first results on the variance of the $GREEDY$ algorithm.

We now review some additional relevant literature.   The
$MIS,MWIS,MM$ and $MWM$ problems are obviously well-studied and
central to the field of combinatorial optimization. The $MIS$
problem is known to be NP-Complete, even for the case of cubic
planar graphs \cite{G.76} and graphs of polynomially large girth
\cite{M.92}, and is known to be $MAX-SNP$ complete even when
restricted to graphs with degree at most 3 \cite{B.94}.   From both
an approximation algorithm  and existential standpoint, the $MIS$
problem has been well-studied for bounded degree graphs
\cite{H.94a}, \cite{H.94b}, \cite{B.94}; graphs with large girth
\cite{M.85}, \cite{M.92}; triangle-free graphs with a given degree
sequence \cite{A.80}, \cite{A.81}, \cite{G.83}, \cite{S.83},
\cite{S.91}; and large-girth graphs with a  given degree sequence,
including regular graphs with large girth \cite{BollobasRandReg},
\cite{HopkinsStaton}, \cite{S.91}, \cite{D.94}, \cite{S.95},
\cite{L.}.  We note that, as already mentioned, our Corollary
~\ref{coro:MICORunweight} was derived earlier in~\cite{ L.} using
different techniques.

Although the $MM$ problem is solvable in polynomial time, much
research has gone into finding specialized algorithms for restricted
families of graphs. The most relevant graph families for which $MM$
has been studied (often using $GREEDY$ and related algorithms) are
bounded-degree graphs, and bounded-degree graphs of girth at least 5
\cite{D.91},\cite{M.97}. However, there appears to be a gap in the
literature for $MM$ in regular graphs with large girth, barring a
recent existential result that an $r$-regular graph with large girth
$g$ always contains a matching of size $\frac{n}{2} -
O((r-1)^{-\frac{g}{2}} n)$~\cite{F.07}. Namely, an asymptotically
perfect matching exists in such graphs as the girth increases. It is
of interest, however, to construct some \emph{decentralized} and
easy to implement algorithm for $MM$ which leads to an
asymptotically perfect matching, and our result
Theorem~\ref{theorem:MM} is a step towards this direction.

Our main method of proof uses the correlation decay technique,
sometimes also called the local weak convergence (objective)
method~\cite{Aldous:assignment00},
\cite{AldousSteele:survey},\cite{GamarnikNowickiSwirscszExpDyn}. We
establish that the choices made by the $GREEDY$ algorithm are
\emph{asymptotically independent} for pairs of nodes (in the case of
independent sets) and edges (in the case of matchings) which are far
apart. That is, if two nodes $i,j$ are at a large graph-theoretic
distance, then $\pr(i,j\in {\cal IG})\approx\pr(i\in {\cal
IG})\pr(j\in {\cal IG})$. A similar statement holds for matchings,
and also for the weighted case with i.i.d. weights. This allows the
reduction of the problem on a graph to the far simpler problem
formulated on a regular tree, which can be solved in a very explicit
way.  Such an asymptotic independence was also observed in
\cite{L.}, but here we are able to characterize this decay in a more
explicit manner.  A similar phenomenon was also observed in
\cite{GamarnikNowickiSwirscszExpDyn}, which studied maximum weight
independent sets and matchings for the case of i.i.d weights in
$r$-regular graphs with girth diverging to infinity. There it was
observed  that for the case of i.i.d. exponentially distributed
weights, such a decay of correlations occurs when $r=3,4$ and does
not occur when  $r \geq 5$, even as the girth diverges to infinity.
Thus the techniques of~\cite{GamarnikNowickiSwirscszExpDyn} were
only able to analyze exponentially weighted independent sets in
regular graphs of large girth when the degree was $r \leq 4$.  In
contrast we show that independent sets produced by $GREEDY$
\emph{always} exhibits such a decay of correlations for any degree.
This allows us to extend the analysis of
\cite{GamarnikNowickiSwirscszExpDyn} to regular graphs of arbitrary
constant degree. In Section~\ref{section:NUMERICS} we will see that
$GREEDY$ is nearly optimal for the settings considered in
\cite{GamarnikNowickiSwirscszExpDyn}.

We now give an outline of the rest of the paper.  In
Section~\ref{section:PRIOR} we state our main results formally and
show that our analysis for the case of i.i.d. weights encompasses
the analysis for the unweighted case. In
Section~\ref{section:BLOCKINGSUBTREES} we introduce the notion of an
influence blocking subgraph, show that under an i.i.d. weighting
most nodes (edges) will belong to such subgraphs, and show that
these subgraphs determine the behavior of $GREEDY$.  This enables us
to prove certain locality properties of $GREEDY$, which we then
apply to the setting of regular graphs of large constant girth.  In
Section~\ref{section:bonussec} we introduce and study a bonus
recursion that we will use to analyze the performance of $GREEDY$ on
infinite $r$-ary trees. Section~\ref{section:varrsec} is devoted to proving
results on the variance of GREEDY.
In Section ~\ref{section:NUMERICS} we
numerically evaluate our bounds and compare to earlier bounds in the
literature.  Finally, in Section~\ref{section:CONCLUSION} we provide
directions for future work and summary remarks.

\subsection{Notations and conventions}
We close this section with some additional notations. Throughout the
paper we consider simple undirected graphs $G=(V,E)$. Given a simple
path $P$ in a graph $G$,  the length of $P$ is the number of edges
in $P$. Given two nodes $i,j\in V$, the distance $D(i,j)$ is the
length of a shortest $i$ to $j$ path in $G$. Similarly, the distance
$D(e_1,e_2)$ between two edges $e_1,e_2\in E$ is the length of the
shortest path in $G$ that contains both $e_1$ and $e_2$, minus one.
Given a node $i\in V$, let the depth$-d$ neighborhood $N_d(i)$ be
the subgraph rooted at $i$ induced by the set of nodes $i'$ with
$D(i,i') \leq d$.  Givn an edge $e$, let $N_d(e)$ denote the
subgraph induced by the set of edges $e'$ with $D(e,e') \leq d$.
Specifically, for every  node $i$ and edge $e$, $N_0(i) = \lbrace i
\rbrace$ and $N_0(e) = \lbrace e \rbrace$. For simplicity we write
$N(\cdot)$ for $N_1(\cdot)$. $|N(i)|$ is the degree of the node $i$,
and $\max_{i\in V}|N(i)|$ is defined to be the degree of the graph.

Given a rooted  tree $T$, the depth of $T$ is the maximum distance
between the root $r$ and any leaf, and the depth of a
 node $i$ in $T$ is  $D(r,i)$. Given a node $i\in T$,
the set of children of $i$ is denoted by $C(i)$.

Suppose the nodes of an undirected graph $G$ are equipped with
weights $W_i$. We say that a path $i_1,i_2,\cdots,i_k$ is node
increasing if $W_{i_1}<\cdots<W_{i_k}$. Similarly, if the edges of
$G$ are weighted $W_{ij}$, we say that a path $i_1,i_2,\cdots,i_k$
is edge increasing if $W_{i_1i_2}<\cdots<W_{i_{k-1}i_k}$.

Denote by $T(r,d), d \geq 1$ a depth-$d$ tree where all non-leaf
nodes have $r$ children, and all leaves are distance $d$ from the
root.  Denote by $T(r+1,r,d), d \ge 1$ the depth-$d$ tree where the
root has $r+1$ children, all other non-leaf nodes have $r$ children,
and all leaves are distance $d$ from the root. Note that if $G \in
G(g,r)$ for some $g \geq 4$, then for every node $i \in V(G)$ and
any $d \leq \lfloor \frac{g-2}{2} \rfloor$, $N_d(i)$ is (isomorphic
to) $T(r,r-1,d)$. By convention, $T(r,0)$ and $T(r+1,r,0)$  both
refer to a single node.

Throughout the paper we will only consider non-negative distribution
functions, so the non-negativity qualification will be implicit.  If
$X$ is a discrete r.v. taking values in $\rZ_+$, the corresponding
probability generating function (p.g.f.) is denoted by
$\phi_X(s)=\sum_{k=0}^\infty s^k \pr(X=k)$. If two r.v. $X$ and $Y$
are equal in distribution, we write $X\stackrel{D}{=}Y$. When $X$ is
distributed according to distribution $F$, we will also write (with
some abuse of notation) $X\stackrel{D}{=}F$. The $m$-fold
convolution of a random variable $X$ is denoted by $X^{(m)}$. Let
$W$ be a continuous r.v. and let $X$ be a r.v. taking non-negative
integer values. Denote by $W^{<X>}$ the r.v. $\max_{1\le i\le X}W_i$
when $X>0$ and $0$ when $X=0$. Here $W_i$ are i.i.d. copies of $W$.
Given two events ${\cal A},{\cal B}$, let ${\cal A}\wedge {\cal B}$
and ${\cal A}\vee {\cal B}$ denote, respectively, the conjunction
and disjunction events. Also ${\cal A}^c$ denotes the complement of
the event ${\cal A}$ and $I({\cal A})$ denotes the indicator
function for the event ${\cal A}$.

\section{Main results} \label{section:PRIOR}

\subsection{Weighted case}
The following is our main result for the performance of the $GREEDY$ algorithm for finding
largest weighted independent sets. Both in the context of independent sets and matchings we assume
that the weights (of the nodes and edges) are generated i.i.d. from a non-negative continuous distribution $F$.
\begin{theorem}\label{theorem:MWISWEIGHT}
For every $g \geq 4$ and $r \geq 3$, and every continuous non-negative r.v. $W\stackrel{D}{=}F$
with density $f$ and $\E[W]<\infty$,
\begin{align*}
\int_{0}^{\infty} &x \Big(r-1 -
(r-2)F(x)\Big)^{-\frac{r}{r-2}}f(x)dx
- \E[W]\frac{r(r-1)^{\lfloor \frac{g-2}{2} \rfloor }}{(\lfloor \frac{g-2}{2} \rfloor+1)!}\\
&\leq \inf_{G \in G(g,r)} \E\left[\frac{W[{\cal IG}]}{|V|}\right] \leq
\sup_{G \in G(g,r)} \E\left[\frac{W[{\cal IG}]}{|V|}\right] \\
&\leq \int_{0}^{\infty} x \Big(r-1 - (r-2)F(x)\Big)^{-\frac{r}{r-2}}f(x)dx
+ \E[W]\frac{r(r-1)^{\lfloor \frac{g-2}{2} \rfloor }}{(\lfloor \frac{g-2}{2} \rfloor+1)!}.
\end{align*}
\end{theorem}
As an immediate corollary, we obtain the following result.
\begin{corollary}\label{coro:MWISWEIGHT}
For every  $r \ge 3$ and every continuous non-negative r.v. $W\stackrel{D}{=}F$
with density $f$ and $\E[W]<\infty$,
\begin{eqnarray}\displaystyle
\lim_{g \rightarrow \infty}\displaystyle \inf_{G \in G(g,r)} \E\left[\frac{W[{\cal IG}]}{|V|}\right]
&=&\nonumber
\lim_{g \rightarrow \infty}\displaystyle \sup_{G
\in G(g,r)} \E\left[\frac{W[{\cal IG}]}{|V|}\right] \\&=&
\int_{0}^{\infty} x \Big(r-1 - (r-2)F(x)\Big)^{-\frac{r}{r-2}}f(x)dx.
\end{eqnarray}
\end{corollary}

We now present the results for matchings.
\begin{theorem}\label{theorem:MWMWEIGHT}
For every $g \ge 4$ and $r \geq 3$, and every continuous non-negative r.v. $W\stackrel{D}{=}F$
with density $f$ and $\E[W]<\infty$,
\begin{align*}
\frac{r}{2}\int_{0}^{\infty} &x \Big(r-1 - (r-2)F(x)\Big)^{-\frac{2(r-1)}{r-2}}f(x)dx
- \E[W]\frac{r (r-1)^{\lfloor \frac{g-2}{2} \rfloor}}{(\lfloor \frac{g-2}{2} \rfloor)!} \\
&\leq \inf_{G \in G(g,r)} \E\left[\frac{W[{\cal MG}]}{|V|}\right]
\leq \sup_{G \in G(g,r)} \E\left[\frac{W[{\cal MG}]}{|V|}\right]\\
&\leq \frac{r}{2}\int_{0}^{\infty} x \Big(r-1 - (r-2)F(x)\Big)^{-\frac{2(r-1)}{r-2}}f(x)dx
+ \E[W]\frac{r (r-1)^{\lfloor \frac{g-2}{2} \rfloor}}{(\lfloor \frac{g-2}{2} \rfloor)!}.
\end{align*}
\end{theorem}

An immediate implication is
\begin{corollary}\label{coro:MWMWEIGHT}
For every  $r \ge 3$ and every continuous non-negative r.v. $W\stackrel{D}{=}F$
with density $f$ and $\E[W]<\infty$,
\begin{eqnarray}\displaystyle \lim_{g \rightarrow \infty}\displaystyle \inf_{G \in G(g,r)} \E\left[\frac{W[{\cal MG}]}{|V|}\right]
&=&\nonumber
\lim_{g \rightarrow \infty}  \displaystyle \sup_{G
\in G(g,r)} \E\left[\frac{W[{\cal MG}]}{|V|}\right] \\&=&
\frac{r}{2} \int_{0}^{\infty} x \Big(r-1 - (r-2)F(x)\Big)^{-\frac{2(r-1)}{r-2}}f(x)dx.
\end{eqnarray}
\end{corollary}

We now state our main results on bounding the variance of $W[{\cal IG}]$ and  $W[{\cal MG}]$.
\begin{theorem}\label{theorem:VAR12}
For every continuous non-negative r.v. $W\stackrel{D}{=}F$ with
 $\E[W^2]<\infty$, and
for every graph $G$ with degree $r \geq 3$,
\begin{equation}\label{VAR12aa}
Var[\frac{W[{\cal IG}]}{|V|}] \leq \frac{ 9 \E[W^2] r^2 e^{(r-1)^3}}{|V|}.
\end{equation}
and
\begin{equation}\label{VAR123aa}
Var[\frac{W[{\cal MG}]}{|E|}] \leq \frac{ 33 E[W^2] r^2 e^{(r-1)^3}}{|E|}.
\end{equation}
\end{theorem}
We stress that, unlike previous results,
no assumption is made on the structure of the graph other than a bound on the maximum degree.

\subsection{Unweighted case}
As we will show in the following subsections, Theorems~\ref{theorem:MWISWEIGHT} and \ref{theorem:MWMWEIGHT}
lead to the following bounds
on the cardinality of independent sets and matchings produced by $GREEDY$ in regular unweighted graphs.

\begin{theorem}\label{theorem:MIS}
For every $g \ge 4$ and $r \geq 3$,
\begin{align}\label{MISa}\nonumber
\frac{1 - (r-1)^{ -\frac{2}{r-2}  } }{2} -\frac{r(r-1)^{\lfloor \frac{g-2}{2} \rfloor }}{(\lfloor \frac{g-2}{2} \rfloor+1)!}\
&\leq \displaystyle \inf_{G \in G(g,r)} \E[\frac{|{\cal IG}|}{|V|}]
\leq
\displaystyle \sup_{G \in G(g,r)} \E[\frac{|{\cal IG}|}{|V|}] \\
&\leq
\frac{1 - (r-1)^{-\frac{2}{r-2}}}{2} + \frac{r(r-1)^{\lfloor \frac{g-2}{2} \rfloor }}{(\lfloor \frac{g-2}{2} \rfloor+1)!}.
\end{align}
\end{theorem}
The following immediate corollary is an analogue of Corollary ~\ref{coro:MWISWEIGHT} for the unweighted case.
\begin{corollary}\label{coro:MICORunweight}
For every $r \ge 3$,
\begin{eqnarray}\displaystyle
\lim_{g \rightarrow \infty}\displaystyle \inf_{G \in G(g,r)} \E\left[\frac{|{\cal IG}|}{|V|}\right]
&=&\nonumber
\lim_{g \rightarrow \infty}  \displaystyle \sup_{G
\in G(g,r)} \E\left[\frac{|{\cal IG}|}{|V|}\right] \nonumber\\
&=&\frac{1 - (r-1)^{-\frac{2}{r-2}}}{2} \nonumber.
\end{eqnarray}
\end{corollary}

A second corollary is the following lower bound on the size of a maximum independent set in an $r$-regular graph
with girth $\ge g$.
\begin{corollary}\label{coro:MICOR}
For every $g \geq 4$ and $r\ge 3$,
\begin{eqnarray}
\displaystyle \inf_{G \in G(g,r)} \frac{|{\cal I}|}{|V|} \geq \frac{1 -
(r-1)^{-\frac{2}{r-2}}}{2} - \frac{r(r-1)^{\lfloor \frac{g-2}{2} \rfloor }}{(\lfloor \frac{g-2}{2} \rfloor+1)!}.
\end{eqnarray}
\end{corollary}

Our results for matchings are as follows.

\begin{theorem}\label{theorem:MM}
For every $g \geq 4$ and $r\ge 3$,
\begin{align}\label{MIS}
\nonumber \frac{1 - (r-1)^{-\frac{r}{r-2}}}{2} - \frac{r (r-1)^{\lfloor \frac{g-2}{2} \rfloor}}{(\lfloor \frac{g-2}{2} \rfloor)!}
&\leq \displaystyle
\inf_{G \in G(g,r)} \E[\frac{|{\cal MG}|}{|V|}]
\leq \displaystyle \sup_{G
\in G(g,r)} \E[\frac{|{\cal MG}|}{|V|}] \\
&\leq \frac{1 -
(r-1)^{-\frac{r}{r-2}}}{2} + \frac{r (r-1)^{\lfloor \frac{g-2}{2} \rfloor}}{(\lfloor \frac{g-2}{2} \rfloor)!}.
\end{align}
\end{theorem}

\begin{corollary}\label{coro:MMCORunweight}
For every $r\ge 3$,
\begin{eqnarray}\displaystyle
\lim_{g \rightarrow \infty}\displaystyle \inf_{G \in G(g,r)} \E\left[\frac{|{\cal MG}|}{|V|}\right]
&=&\nonumber
\lim_{g \rightarrow \infty}  \displaystyle \sup_{G
\in G(g,r)} \E\left[\frac{|{\cal MG}|}{|V|}\right] \nonumber\\
&=&\frac{1 - (r-1)^{-\frac{r}{r-2}}}{2} \nonumber.
\end{eqnarray}
As a result
\begin{align*}
\lim_{r\rightarrow\infty}
\lim_{g \rightarrow \infty}\displaystyle \inf_{G \in G(g,r)} \E\left[\frac{|{\cal MG}|}{|V|}\right]
&=\nonumber
\lim_{r\rightarrow\infty}
\lim_{g \rightarrow \infty}  \displaystyle \sup_{G
\in G(g,r)} \E\left[\frac{|{\cal MG}|}{|V|}\right] \nonumber\\
&=\frac{1}{2} \nonumber.
\end{align*}
\end{corollary}
Namely, $GREEDY$ finds a nearly perfect matching when both the degree and girth are large.
A second corollary is the following lower bound on the size of a maximum matching in an $r$-regular graph
with girth $\ge g$.
\begin{corollary}\label{coro:MMCOR}
For every $g \geq 4$ and $r \geq 3$,
\begin{eqnarray}
\inf_{G \in G(g,r)}\frac{|{\cal M}|}{|V|} \geq  \frac{1 -
(r-1)^{-\frac{r}{r-2}}}{2} - \frac{r (r-1)^{\lfloor \frac{g-2}{2} \rfloor}}{(\lfloor \frac{g-2}{2} \rfloor)!}.
\end{eqnarray}
\end{corollary}

Bounds on the variance of $W[{\cal IG}],W[{\cal MG}]$ will result in the following bounds for the variance
of $|{\cal IG}|,|{\cal MG}|$.
\begin{theorem}\label{theorem:VAR11}
For every graph $G$ with degree $r \geq 3$,
\begin{equation}\label{eq:VAR1ff}
Var[\frac{|{\cal IG}|}{|V|}] \leq \frac{ 9 r^2 e^{(r-1)^3}}{|V|}.
\end{equation}
and
\begin{equation}\label{eq:VAR1ff2}
Var[\frac{|{\cal MG}|}{|E|}] \leq \frac{ 33 r^2 e^{(r-1)^3}}{|E|}.
\end{equation}
\end{theorem}

\subsection{Converting the weighted case to the unweighted case}\label{section:weightme}
In this section, we prove that all of the results pertaining to
$GREEDY$'s performance w.r.t. finding unweighted  independent sets
and matchings are implied by our analysis for the case of i.i.d.
weights.  This will allow us to focus only on the case of i.i.d.
weights for the remainder of the paper.
\begin{lemma}\label{weightimpy}
Theorem~\ref{theorem:MWISWEIGHT} implies Theorem~\ref{theorem:MIS} and Theorem~\ref{theorem:MWMWEIGHT}
implies Theorem~\ref{theorem:MM}.
\end{lemma}
\begin{proof}
We first prove that Theorem~\ref{theorem:MWISWEIGHT} implies Theorem~\ref{theorem:MIS}.
Fix $\epsilon>0$. Let $F$ be a uniform distribution on $[1-\epsilon,1+\epsilon]$.
Applying Theorem~\ref{theorem:MWISWEIGHT} we have
\begin{align*}
\int_{1-\epsilon}^{1+\epsilon} &x ((r-1) - (r-2)\frac{x - (1-\epsilon)}{2\epsilon})^{-\frac{r}{r-2}}\frac{1}{2\epsilon}dx
- \frac{r(r-1)^{\lfloor \frac{g-2}{2} \rfloor }}{(\lfloor \frac{g-2}{2} \rfloor+1)!} \\
&\leq \inf_{G \in G(g,r)} \E\left[\frac{W[{\cal IG}]}{|V|}\right] \leq
\sup_{G \in G(g,r)} \E\left[\frac{W[{\cal IG}]}{|V|}\right] \\
&\leq
\int_{1-\epsilon}^{1+\epsilon} x ((r-1) - (r-2)\frac{x -
(1-\epsilon)}{2\epsilon})^{-\frac{r}{r-2}}\frac{1}{2\epsilon}dx +
\frac{r(r-1)^{\lfloor \frac{g-2}{2} \rfloor }}{(\lfloor \frac{g-2}{2} \rfloor+1)!}.
\end{align*}

Note that  $(1 -\epsilon) \E[|{\cal IG}|] \leq  \E[W[{\cal IG}]]
\leq (1 + \epsilon) \E[|{\cal IG}|]$,
and for $1-\epsilon \leq x \leq
1+\epsilon$ we have $((r-1) - (r-2)\frac{x - (1
-\epsilon)}{2\epsilon})^{-\frac{r}{r-2}}\frac{1}{2\epsilon} \geq 0$.
Thus
\begin{align*}
(1-\epsilon)\int_{1-\epsilon}^{1+\epsilon} &((r-1) - (r-2)\frac{x - (1-\epsilon)}{2\epsilon})^{-\frac{r}{r-2}}\frac{1}{2\epsilon}dx
- \frac{r(r-1)^{\lfloor \frac{g-2}{2} \rfloor }}{(\lfloor \frac{g-2}{2} \rfloor+1)!}\\
&\leq (1 + \epsilon) \inf_{G \in G(g,r)} \E\left[\frac{|{\cal IG}|}{|V|}\right] \leq
 (1 + \epsilon) \sup_{G \in G(g,r)} \E\left[\frac{|{\cal IG}|}{|V|}\right] \\
&\leq \frac{(1+\epsilon)^2}{1-\epsilon}\int_{1-\epsilon}^{1+\epsilon}
((r-1) - (r-2)\frac{x - (1-\epsilon)}{2\epsilon})^{-\frac{r}{r-2}}\frac{1}{2\epsilon}dx
+ \frac{1+\epsilon}{1-\epsilon} \frac{r(r-1)^{\lfloor \frac{g-2}{2} \rfloor }}{(\lfloor \frac{g-2}{2} \rfloor+1)!}.
\end{align*}
Letting $u = \frac{x - (1-\epsilon)}{2 \epsilon}$, we can apply integration by substitution to find that:
\begin{align*}
(1-\epsilon)\int_{0}^{1} &((r-1) - (r-2)u)^{-\frac{r}{r-2}}du
- \frac{r(r-1)^{\lfloor \frac{g-2}{2} \rfloor }}{(\lfloor \frac{g-2}{2} \rfloor+1)!}\\
&\leq (1 + \epsilon) \inf_{G \in G(g,r)} \E\left[\frac{|{\cal IG}|}{|V|}\right]
\leq (1 + \epsilon) \sup_{G \in G(g,r)} \E\left[\frac{|{\cal IG}|}{|V|}\right]\\
&\leq \frac{(1+\epsilon)^2}{1-\epsilon} \int_{0}^{1} ((r-1) -
(r-2)u)^{-\frac{r}{r-2}}du + \frac{1+\epsilon}{1-\epsilon}
\frac{r(r-1)^{\lfloor \frac{g-2}{2} \rfloor }}{(\lfloor \frac{g-2}{2} \rfloor+1)!}.
\end{align*}
Evaluating the integrals and letting $\epsilon \rightarrow 0$ then
demonstrates the desired result.  The proof that  Theorem~\ref{theorem:MWMWEIGHT} implies
Theorem~\ref{theorem:MIS} follows identically,
using the bounds for matchings instead of those for independent
sets.
\end{proof}

\begin{lemma}\label{lemma:weightimpy2}
Theorem~\ref{theorem:VAR12} implies Theorem~\ref{theorem:VAR11}.
\end{lemma}
\begin{proof}
We first prove that (\ref{VAR12aa}) implies (\ref{eq:VAR1ff}).
Let again $F$ be a uniform distribution on $[1-\epsilon,1+\epsilon]$.
The following bounds are immediate
\begin{align*}
(1 - \epsilon)^2 &\E[(\frac{|{\cal IG}|}{|V|})^2] - (1 +\epsilon)^2 \E^2[\frac{|{\cal IG}|}{|V|}] \\
&\leq Var[\frac{W[{\cal IG}]}{|V|} ] \\
&\leq (1 + \epsilon)^2
\E[(\frac{|{\cal IG}|}{|V|})^2] - (1 - \epsilon)^2
\E^2[\frac{|{\cal IG}|}{|V|}],
\end{align*}

which implies that
$$| Var[ \frac{|{\cal IG}|}{|V|} ] - Var[ \frac{W[{\cal IG}]}{|V|} ]|
\leq (2 \epsilon + \epsilon^2) ( \E[(\frac{|{\cal
IG}_W|}{|V|})^2] + \E^2[\frac{|{\cal IG}|}{|V|}]).$$
Thus since the second moment of $F$ is $1 +\frac{\epsilon^2}{3}$, by the triangle inequality and
Theorem~\ref{theorem:VAR12} we find that for any
graph $G$ of maximum degree $r$,
$$
Var[ \frac{|{\cal IG}|}{|V|} ] \leq \frac{ 9 (1 +
\frac{\epsilon^2}{3})  r^2 e^{(r-1)^3}}{|V|} + (2 \epsilon +
\epsilon^2) ( \E[(\frac{|{\cal IG}|}{|V|})^2] +
\E^2[\frac{|{\cal IG}|}{|V|}]).$$ Observing that
$|\frac{|{\cal IG}|}{|V|}| \leq 1$, we see that (\ref{eq:VAR1ff})
follows by  letting $\epsilon \rightarrow 0$.

The proof of (\ref{eq:VAR1ff2}) from  (\ref{VAR123aa}) is done similarly.
\end{proof}

\section{Influence blocking subgraphs}\label{section:BLOCKINGSUBTREES}
In this section we introduce the notion of an influence blocking
subgraph, and give a useful characterization of these
subgraphs.  We then bound the probability that a node (edge)  of
a bounded degree graph $G$ is contained in (an appropriately) small influence blocking subgraph under an
i.i.d. weighting from any continuous distribution function. Throughout this section we consider
a graph whose nodes and edges are equipped with non-negative distinct (non-random unless otherwise stated)
weights $W_i, i\in V$ and  $W_{e}, e\in E$.

\begin{definition}\label{defi:ibtreedef}
A subgraph $H$ of $G$
is called an influence blocking subgraph  (i.b.s.) if
for every node (edge) $z \in H, W_z > \max_{y \in N(z) \setminus H} W_y$.
\end{definition}
Here $ N(z) \setminus H$ means the set of nodes or edges (depending on the context) in $N(z$)
which do not belong to $H$.
We now show that for any set of nodes (edges) $Z$ there exists a unique minimal i.b.s. $H$
containing $Z$,
and give a simple characterization of this subgraph.
\begin{lemma}\label{lemma:quiteacharacter}
Given a set of nodes (edges) $Z$ there exists a unique minimal i.b.s. $H$ containing $Z$.
Namely, for every other i.b.s. $H'$ containing $Z$,
$H$ is a subgraph of $H'$. Moreover, $H$ is characterized as the set of nodes (edges) $z$ such
that there exists a node (edge) increasing path $z_1,\ldots,z_k$ with $z_1\in Z$ and $z_k=z$.
\end{lemma}
We denote this unique minimal i.b.s.  by $IB_G(Z)$, or $IB(Z)$ when the underlying graph $G$ is unambiguous.

\begin{proof}
We first show that $IB_G(Z)$ is contained in every i.b.s. $T$ containing $Z$.
Suppose, for the purposes of contradiction, there exists an increasing path
$z_1,z_2,\ldots,z_k$ such that $z_1\in Z,z_k\notin T$. Let $l < k$ be the largest index
such that $z_l\in T$. Then $z_{l+1}\in N(z_l)$, but $W_{z_{l+1}}>W_{z_l}$, which is a contradiction
to the fact that $T$ is an i.b.s.

We now show that $IB(Z)$ is itself an i.b.s. containing $Z$. By definition $Z\subset IB(Z)$.
Now let $z\in IB(Z)$ be arbitrary and let $z'\in N(z)\setminus IB(Z)$ be arbitrary as well.
If $W_{z'}>W_z$, then since there exists an increasing path from $Z$ to $z$, by appending $z'$ to this path
we obtain an increasing path from $Z$ to $z'$ and thus $z'\in IB(Z)$, which is a contradiction.
We conclude $W_{z'}<W_z$, and the proof is complete.
\end{proof}

We now show that the existence of a `small' i.b.s. for $N(v) (\ N(e)\ )$ is independent of $W_v (W_e)$ under an i.i.d. weighting.
\begin{lemma}\label{lemma:ind11}
Given an arbitrary node (edge) $z$,  $IB(N(z))\subset N_d(z)$
holds iff there does not exist a node (edge) increasing path between some node (edge) $z'\in N(z)$ and
$z''\in N_{d+1}(z)\setminus N_{d}(z)$ which is contained entirely in $G\setminus z$.
As a result, if the node (edge) weights of $G$ are generated i.i.d. from a continuous distribution $F$, then
the event $IB(N(z))\subset N_d(z)$ is independent from $W_z$.
\end{lemma}
\begin{proof}
If there exists an increasing path $z_1,\ldots,z_k$ between $N(z)$ and $N_{d+1}(z)\setminus N_{d}(z)$,
then the last element $z_k\in N_{d+1}(z)\setminus N_{d}(z)$
must belong to $IB(N(z))$ and thus $IB(N(z))\subset N_d(z)$ cannot hold.
Now suppose no increasing path exists between $N(z)$ and $N_{d+1}(z)\setminus N_{d}(z)$
inside $G\setminus z$. Then no increasing path can exist between $N(z)$ and $N_{d+1}(z)\setminus N_{d}(z)$
inside $G$ either, since in any such path we can find a subpath which does not use $z$. This completes the
proof of the first part of the lemma. The second part is an immediate implication.
\end{proof}

The usefulness of the i.b.s. comes from the following lemma, which informally states
that the decisions taken by $GREEDY$ inside an i.b.s. $H$ are not affected by the complement of $H$
in $G$.

\begin{lemma}\label{lemma:SBOUNDEDSAYSITALL}
Suppose $H$ is an i.b.s. of $G$.
Then ${\cal IG}(G)\cap V(H)={\cal IG}(H)$ (${\cal MG}(G)\cap E(H)={\cal MG}(H)$),
where the weights of $H$ are induced from $G$.
\end{lemma}
\begin{proof}
Let $z_1,z_2,\ldots,z_m$ be the nodes (edges) of $H$ ordered in decreasing order by their weight.
We show by induction in $k=1,2,\ldots,m$ that $z_k\in {\cal IG}(G) (z_k\in {\cal MG}(G))$ iff
$z_k\in {\cal IG}(H) (z_k\in {\cal MG}(H))$. For the base case $k=1$ observe that $z_1$ is the heaviest
element of $H$. Since $H$ is an i.b.s. then also $z_1$ cannot have a heavier neighbor in $G\setminus H$.
Thus $GREEDY$ will select it both for $G$ and $H$.

We now prove the induction step and assume the assertion holds for all $k'\le k-1 < m$. Suppose $z_k$
was not accepted by  $GREEDY$ when it was operating on $G$. This means that $GREEDY$ accepted
some neighbor of $z_k$ which was heavier than $z_k$ and, as a result, deleted $z_k$. Since
$H$ is an i.b.s. this neighbor must be in $H$, namely it is $z_{k'}$ for some $k'<k$.
By the inductive assumption $GREEDY$ selected $z_{k'}$ when it was operating on $H$ as well.
Then all neighbors of $z_k'$ in $H$ are deleted including $z_k$, and thus $z_k$ cannot be accepted by $GREEDY$ when operating on $H$. Similarly, suppose
$GREEDY$ did not select $z_k$ when it was operating on $H$. Namely, $GREEDY$ accepted some
neighbor $z_{k'}$ of $z_k$ with $k'<k$. By the inductive assumption the same holds for $GREEDY$
operating on $G$: $z_{k'}$ was accepted and all neighbors, including $z_k$ were deleted. This completes
the proof of the induction step.
\end{proof}

We now  bound the  probability that $IB(N(z))$
is contained in $N_d(z)$ when $z$ is a node (edge) in a bounded degree
graph $G$ and the weights are random.

\begin{lemma}\label{lemma:shortandsweet}
Let $G$ be any graph of maximum degree $r \geq 3$, and suppose that the nodes  and edges of $G$
are equipped with i.i.d. weights from a continuous distribution $F$.
Then for any node (edge) $i(e)$ and any $d \geq 0$,
\begin{align*}
\pr(IB(N(i))\subset N_d(v)) &\geq 1 - \frac{r(r-1)^{d}}{(d+1)!},\\
\pr(IB(N(e))\subset N_d(e)) &\geq 1 - \frac{2(r-1)^{d+1}}{(d+1)!},
\end{align*}
where the first (second) inequality is understood in the context of node (edge) weights.
\end{lemma}
\begin{proof}
Any length-$k$ path equipped with i.i.d. node (edge) weights
generated using a continuous distribution is a node (edge)
increasing path with probability  equal to $1/(k+1)!\ (\ 1/k!\ )$.
For every node $z \in G$ there exist at most $r(r-1)^{d}$ distinct
length$-d$ paths in $G\setminus z$ that originate on some node in
$N(z) \setminus z$ and use exactly one node from $N(z)$. For every
edge $z \in G$, there exist at most $2(r-1)^{d+1}$ distinct
length$-(d+1)$ paths in $G\setminus z$ that originate on some edge
in $N(z) \setminus z$ and use exactly one edge from $N(z)$.
Observe that every node  increasing
path originating in $N(z) \setminus z$
and terminating in $N_{d+1}(z) \setminus N_d(z)$ must contain a length$-d$  node
increasing subpath originating  in $N(z)\setminus z$ which uses exactly one node  of $N(z)$.
We then obtain the result by applying a simple union bound and Lemma~\ref{lemma:ind11}.
\end{proof}

We now state and prove the main result of this section.
\begin{theorem}\label{theorem:GraphTreeRelation}
Let $G\in G(g,r)$ for some $g \geq 4$ , and $d \geq \lfloor \frac{g-2}{2} \rfloor$ be arbitrary.  Let $T=T(r,r-1,d)$ have root $0$.
Suppose the nodes and edges of $G$ and $T$ are equipped with i.i.d.
weights from a continuous distribution $F$. Then for every node $i\in V(G)$, edge $e\in E(G)$,
and every child $j$ of $0$ in $T$
\begin{align}
&\Big|\E[W_iI(i\in {\cal IG}(G))]-\E[W_0I(0\in {\cal IG}(T))]\Big|
\le \E[W]\frac{r(r-1)^{\lfloor \frac{g-2}{2} \rfloor }}{(\lfloor \frac{g-2}{2} \rfloor+1)!} \label{eq:WgraphWtreeIS}
\end{align}
and
\begin{align}
&\Big|\E[W_{e}I(e\in {\cal MG}(G))]-\E[W_{0j}I((0,j)\in {\cal MG}(T))]\Big|
\le \E[W]\frac{2(r-1)^{\lfloor \frac{g-2}{2} \rfloor}}{(\lfloor \frac{g-2}{2} \rfloor)!}, \label{eq:WgraphWtreeM}
\end{align}
where $W\stackrel{d}{=}F$.
Also the limits
\begin{align}
\lim_{d\rightarrow\infty}\pr(0\in {\cal IG}(T)), ~~
\lim_{d\rightarrow\infty}\pr((0,j)\in {\cal MG}(T)) \label{eq:limits}
\end{align}
exist.
\end{theorem}
\remark It is important to note that the bounds of this theorem hold for \emph{any} value of $d \geq \lfloor \frac{g-2}{2} \rfloor$.
It is this property which will ultimately lead to the existence of limits (\ref{eq:limits}), as we will see
shortly in the proof. Later on the existence of these limits will lead to a simple expression for the limiting value
of $\E[W_iI(i\in {\cal IG}(T))]$ and $\E[W_{e}I(e\in {\cal MG}(T))]$.

\begin{proof}
Denote $IB(N(i))$ with respect to $G$ by $H(i)$ and $IB(N(0))$ with respect to $T$ by $H(0)$ for simplicity.
Let $d_0\triangleq \lfloor \frac{g-2}{2} \rfloor \leq d$. Then $N_{d_0}(i)$ is a $T(r,r-1,d_0)$ tree. We can construct a coupling in which
 $T=T(r,r-1,d)$ is the natural
extension of this tree with additional node weights generated independently from the node weights of $G$.
In this setting the node $i$ takes the role of the root $0$ of $T$.
We have
\begin{align*}
W_iI(i\in {\cal IG}(G))&=W_iI(i\in {\cal IG}(G),H(i)\subset N_{d_0}(i))
+W_iI(i\in {\cal IG}(G),H(i)\not\subset N_{d_0}(i))\\
&=W_0I(0\in {\cal IG}(T),H(0)\subset N_{d_0}(0))
+W_iI(i\in {\cal IG}(G),H(i)\not\subset N_{d_0}(i)),
\end{align*}
where the second equality follows from Lemma~\ref{lemma:SBOUNDEDSAYSITALL}. This sum is upper bounded by
\begin{align*}
\le W_0I(0\in {\cal IG}(T))+W_iI(H(i)\not\subset N_{d_0}(i)).
\end{align*}
It follows that
\begin{align*}
\E[W_iI(i\in {\cal IG}(G))-W_0I(0\in {\cal IG}(T))]&\le \E[W_iI(H(i)\not\subset N_{d_0}(i))]\\
&=\E[W]\pr(H(i)\not\subset N_{d_0}(i)) \\
&\le \E[W]\frac{r(r-1)^{d_0}}{(d_0+1)!},
\end{align*}
where the equality follows from the second part of Lemma~\ref{lemma:ind11} and the last inequality follows
from Lemma~\ref{lemma:shortandsweet}.
We complete the proof of the  bound (\ref{eq:WgraphWtreeIS}) by establishing a similar bound
with the roles of $W_iI(i\in {\cal IG}(G))$ and $W_0I(0\in {\cal IG}(T))$ reversed.

We now establish the last part of the theorem, namely the existence of limits (\ref{eq:limits}).
Consider any $d'>d$. Let $T'=T(r,r-1,d')$ be a natural extension of the tree $T$
with the same root $0$. Namely, the additional nodes of $T'$ are weighted i.i.d. using $F$,
independently from the weights of the nodes already in $T$. Let  $H'$ denote $IB(N(0))$ with respect to $T'$.
We have
\begin{align*}
I(0\in {\cal IG}(T))&=I(0\in {\cal IG}(T),H'\subset T)+I(0\in {\cal IG}(T),H'\not\subset T) \\
&=I(0\in {\cal IG}(T'),H'\subset T)+I(0\in {\cal IG}(T),H'\not\subset T) \\
&\le I(0\in {\cal IG}(T'))+I(H'\not\subset T).
\end{align*}
This implies that
\begin{align*}
\pr(0\in {\cal IG}(T))-\pr(0\in {\cal IG}(T'))&\le \pr(H'\not\subset T) \\
&\le r(r-1)^d/(d+1)!,
\end{align*}
where the last inequality follows from Lemma~\ref{lemma:shortandsweet}. By reversing the roles of $T$ and $T'$
we obtain
\begin{align*}
\pr(0\in {\cal IG}(T'))-\pr(0\in {\cal IG}(T))\le r(r-1)^d/(d+1)!.
\end{align*}
We conclude that the sequence $\pr(0\in {\cal IG}(T(r,r-1,d))), d\ge 1$ is Cauchy and therefore has a limit.
This concludes the proof for the case of independent sets.
The proof for the case of matchings is obtained similarly and is omitted.
\end{proof}

\section{Bonus, bonus recursion and proofs of the main results}\label{section:bonussec}

\subsection{Bonus and bonus recursion}
 In this subsection, we introduce the notion of a bonus for independent sets and
matchings on trees. Consider a tree $T$ with root $0$, whose nodes (edges) are equipped with distinct positive weights
  $W_i, i\in T (W_{i,j}, (i,j)\in T)$.
\begin{definition}\label{defi:SBONUSDEF}
For every node $i \in T$ let
\begin{align*}
S(i) & = \begin{cases} W_i & \text{if i is a leaf;}\\
W_i I(W_i > \max_{j \in C(i)}S(j)) & \text{otherwise;}
\end{cases}\\
MS(i) &= \begin{cases} 0 & \text{if i is a leaf;}\\
\max_{j\in C(i)}(W_{ij}\ I( W_{ij} > MS(j))) & \text{otherwise;} \end{cases}
\end{align*}
\end{definition}
The quantities $S(i), MS(i)$ are called the bonus of $i$ in the rooted tree $T$
and will be used for the analysis of independent sets and matchings respectively.
Let $T_i$ be the subtree of $T$ rooted at $i$.
Note that the bonus of $i$ depends only on the subtree $T_i$.  To avoid ambiguity, for a subtree $H$ of $T$ rooted at $i$ we let $MS_H(i)$ denote the bonus of $i$ computed w.r.t. the subtree $H$.  We now prove that $S(i)(MS(i))$ determines whether the root $0$ belongs to ${\cal IG}(T) ({\cal MG}(T))$.
\begin{proposition}\label{prop:SUPERGOODENOUGH}
Given a weighted rooted tree $T$ with distinct positive weights on the nodes and edges, for
every node   $i$ and edge  $(i,j)$,
\begin{enumerate}
\item \emph{[Independent sets]}. $S(i)  = W_i I (i \in {\cal IG}(T_i))$.
Specifically, for the root $0$ we obtain $S(0)=W_0 I (0 \in {\cal IG}(T))$.
\item \emph{[Matchings]}.
$
MS(i) = \max_{j \in C(i)} W_{ij} I ((i,j) \in {\cal MG}( T_i ) ).
$
Specifically, for the root $0$ we obtain
$
MS(0)=\max_{j \in C(0)} W_{0j} I ((0,j) \in {\cal MG}( T ) ).
$
\item \emph{[Matchings]}. For every $j\in C(i)$, $(i,j)\in {\cal MG}(T_i)$ iff
$W_{ij}>\max(MS(j), MS_H(i))$, where $H$ is the subgraph of $T_i$ obtained by deleting $(i,j)\cup T_j$.
\end{enumerate}
\end{proposition}
\begin{proof}
Let $d$ be the depth of $T$.
We first prove part 1. The proof proceeds by induction on the depth of a node, starting from nodes at depth $d$.
Thus for the base case, suppose $i$ belongs to level $d$ of $T$, and, as a result, it is a leaf. Then
$S(i)=W_i$. On the other hand, $T_i=\{i\}$ and $I(i \in {\cal IG}(T_i)) = 1$, and the claim follows.

For the induction part assume that the hypothesis is true for all nodes at depth $\geq k+1$ for  $k\le d-1$.
Let $i$ be some node at depth $k$.  Observe that $GREEDY$ selects
node $i$ for inclusion in ${\cal IG}(T_i)$ iff $i$ is not adjacent
to any nodes in $T_i$ that are selected by $GREEDY$ prior to node
$i$ being examined by $GREEDY$.  The  set of nodes in $T_i$ examined by
$GREEDY$ before $i$ are those nodes $j$ such that $W_j > W_i$.  Thus the event
$i \in {\cal IG}(T_i)$ occurs iff for all $j \in C(i)$ s.t.
$W_j > W_i$, we have $j \notin {\cal IG}(T_i)$. We claim that for each such $j$,
$j \notin {\cal IG}(T_i)$ iff $j \notin {\cal IG}(T_j)$. Indeed,
the event $j \notin {\cal IG}(T_i)$ is determined by a subgraph $H$ of $T_i$ induced by nodes with
weights at least $W_j$. Therefore this subgraph does not include $i$ if $W_j>W_i$. It follows that
$H\cap T_j$ is disconnected from the rest of $H$ and then the claim follows.

We conclude that $i\in {\cal IG}(T_i)$ iff
for each $j\in C(i)$ either $W_i>W_j$, or  $j \notin {\cal IG}(T_j)$.
Combining, $i\in {\cal IG}(T_i)$ iff $W_i>\max_{j\in  C(i)}W_j I(j \in {\cal IG}(T_j))$, but by the inductive hypothesis,
$W_j I(j \in {\cal IG}(T_j))=S(j)$. Therefore, $i\in {\cal IG}(T_i)$ iff
$W_i>\max_{j\in  C(i)} S(j)$ and the inductive assertion follows.

We now prove part 2.
The proof is again by induction on the depth of a node.
Base case: $i$ is at lowest depth $d$ and thus a leaf.  In this case, $C(i) = \emptyset$,
and thus $\max_{j \in C(i)} W_{ij} I (W_{ij} \in {\cal MG}( T_j ) ) = MS(i) = 0$.
For the induction step
assume that the induction hypothesis is true for all nodes at depth $\geq k+1, k\le d-1$.
Let $i$ be some node at depth $k$. If $i$ is a leaf we use the same argument as for the base case.
Thus assume $i$ is  not a leaf. Suppose $(i,j_1)\in {\cal MG}(T_i)$.
We claim that then $W_{ij_1}>MS(j_1)$. Indeed
Observe that $GREEDY$ selects $(i,j_1)$  for inclusion in ${\cal MG}(T_i)$ iff $(i,j_1)$ is not
itself adjacent to any edges in $T_i$ that are selected by $GREEDY$
prior to $(i,j_1)$ being examined by $GREEDY$. Thus the event $(i,j_1) \in {\cal MG}(T_i)$
implies that
for all $l \in C(j_1)$ s.t.
$W_{j_1,l} > W_{ij_1}$, we have $(j_1,l) \notin {\cal MG}(T_i)$.
Repeating the argument used for the case of independent sets, we claim that the event
$(j_1,l) \notin {\cal MG}(T_i)$ occurs iff the event $(j_1,l) \notin {\cal MG}(T_{j_1})$ occurs.
Therefore, the event $(i,j_1) \in {\cal MG}(T_i)$ implies that for each $l\in C(j_1)$ either
$W_{ij_1}>W_{j_1,l}$ or $(j_1,l) \notin {\cal MG}(T_{j_1})$, namely the event
$W_{ij_1}>\max_{l\in C(j_1)}W_{j_1,l}I((j_1,l) \in {\cal MG}(T_{j_1}))$ occurs,
which by induction hypothesis is equivalent to the event
$W_{ij_1}>MS(j_1)$, as claimed.

We now complete the proof of the induction step. First assume that $W_{ij}<MS(j)$ for all $j\in C(i)$.
Then from the preceding claim we obtain that no edge $(i,j)$ belongs to ${\cal MG}(T_{i})$ and the claim
is established. Otherwise, let $j_1\in C(i)$ be such that $W_{ij_1}$ is the largest weight
among edges $W_{ij}, j\in C(i)$ satisfying $W_{ij}>MS(j)$.
By the choice of $j_1$ and the preceding claim it follows
that if $W_{ij'}>W_{ij_1}$, then $(i,j')\notin {\cal MG}(T_{i})$. Thus it remains to show
that $(i,j_1)\in {\cal MG}(T_{i})$. The $GREEDY$ examines $(i,j_1)$ after edges $(i,j)$
with $W_{ij}>W_{ij_1}$, but before edges $(i,j)$ with $W_{ij}<W_{ij_1}$. Since edges with $W_{ij}>W_{ij_1}$
were rejected, then whether $(i,j_1)$ is accepted
is determined completely by $(i,j_1)$ plus the subtree $T(j_1)$. Repeating the argument above, we see
that $(i,j_1)$ is accepted iff
$W_{ij_1}>\max_{l\in C(j_1)}W_{j_1,l}I((j_1,l) \in {\cal MG}(T_{j_1}))$, which, by the inductive
hypothesis occurs iff $W_{ij_1}>MS(j_1)$, which is satisfied by the choice of $j_1$.

To prove part 3, we repeat the arguments used to prove parts 1 and 2 to observe that the $GREEDY$
selects $(i,j)$ iff for all neighbors $l$ of $j$ in $T(j)$ with $W_{j,l}>W_{ij}$, the edge $(j,l)$ is rejected
by $GREEDY$ in $T(j)$, and for all neighbors $l\ne j$ of $i$ in $T(i)\setminus ((i,j)\cup T(j))$,
with $W_{il}>W_{ij}$, the edge $(i,l)$ is rejected by $GREEDY$ in $T(i)\setminus ((i,j)\cup T(j))$.
\end{proof}

\subsection{Distributional recursion for bonuses}
We now introduce two sequences of recursively defined
random variables $\lbrace X_{d,r} \rbrace , d\ge 0,$ and $\lbrace Y_{d,r} \rbrace ,d\ge 0$ for any given integer $r\ge 2$. These sequences will play a key
role in understanding the probability distribution of the bonuses introduced in the previous subsection.

Given a positive integer $k$, let $B(k)$
denote a Bernoulli random variable with $\pr(B(k)=1)=1/k$.
Define
\begin{align}
X_{d,r} &\stackrel{D}{=} \begin{cases} 1 & d = 0; \\
(X_{d-1,r}^{(r)} + 1)B({{X_{d-1,r}^{(r)}} + 1}) & d \geq 1; \label{eq:recursionX}
\end{cases}\\
Y_{d,r} &\stackrel{D}{=} \begin{cases} 0 & d = 0;\\
\Big((Y_{d-1,r} + 1)B({{Y_{d-1,r}} + 1})\Big)^{(r)} & d \geq 1; \label{eq:recursionY}
\end{cases}
\end{align}
For an integer-valued r.v. $Z \geq 1$, the joint probability distribution of $Z$, $B(Z)$ is assumed to be
$\pr(Z=z,B=1)=(1/z)\pr(Z=z)$.

It is immediate from these recursions that for all $d\ge 1$
\begin{align}
\E[X_{d,r}]=1, ~~\E[Y_{d,r}]=r. \label{eq:expectationXY}
\end{align}

In the following lemma we show that the distribution of the bonuses $S$ and $MS$ on regular trees
have a very simple representation in terms of the constructed sequences $\lbrace X_{d,r} \rbrace , \lbrace Y_{d,r} \rbrace$.
\begin{lemma}\label{lemma:WBONUSDIST}
Suppose the nodes  and edges of a tree $T(r,d)$ with root $0$ are equipped with i.i.d.
weights generated according to a continuous distribution $F$. Then
$S(0)\stackrel{D}{=} W^{<X_{d,r}>}$,  and $MS(0)\stackrel{D}{=} W^{<Y_{d,r}>}$,
where $W\stackrel{D}{=}F$.
\end{lemma}

\begin{proof}We first prove the identity for $S(0)$.  The proof proceeds by induction on $d$.
For the base case, suppose $d = 0$.  Then $S(0) \stackrel{D}{=} W$ and the conclusion trivially holds.
For the  induction step, assume the hypothesis is true for all $d' < d$.  Let $T_j$ denote the depth $d-1$
subtree of $T(r,d)$ rooted at the $j$-th child of $0$. By the inductive hypothesis we have that
$S(j)$ is distributed as $W^{<X_{d-1,r}>}$.
This implies that $\max_{j \in C(i)} S(j)$ is distributed as
$(W^{<X_{d-1,r}>})^{<r>}$.  Since $W_{0}$ is drawn independent of $\max_{j \in C(i)} S(j)$, we have that
$S(0) \stackrel{D}{=} W_0 I (W_0 > (W^{<X_{d-1,r}>})^{<r>})$. The event underlying $I(\cdot)$ means
that $W_0$ is the largest among $K+1$ random variables distributed according to $F$, where
$K\stackrel{D}{=}X_{d-1,r}^{(r)}$. The required identity then follows from the definition of $X_{d,r}$.

We now establish the identity for $MS(0)$ using induction in $d$.
For the base case $d = 0$ we have  $MS(0)=0$ and the conclusion trivially holds.
For the induction case, assume that the hypothesis is true for all $d' < d$.  Let again $T_j$ denote the depth $d-1$
subtree of $T(r,d)$ rooted at the $j$-th child of $0$.  By the inductive hypothesis
$MS(j)$ is distributed as $W^{<Y_{d-1,r}>}$. Since
$W_{0j}$ is drawn i.i.d., we have
$MS(0) \stackrel{D}{=}
\max_{j \in C(0)} W_{0j} I (W_{0j} > W^{<Y_{d-1,r}>})$, where $W_{0j} I (W_{0j} > W^{<Y_{d-1,r}>})$ is independent for each $j$.
Note that for each $j$, $W_{0j} I (W_{0j} > W^{<Y_{d-1,r}>})$ is by definition distributed as the maximum of $(Y_{d-1,r} + 1)B(Y_{d-1,r} + 1)$ i.i.d. realizations of $W$.  Thus $\max_{j \in C(0)} W_{0j} I (W_{0j} > W^{<Y_{d-1,r}>})$ is distributed as the maximum of $r$ independent samples of the maximum of $(Y_{d-1,r} + 1)B(Y_{d-1,r} + 1)$ i.i.d. realizations of $W$, which by the basic properties of maxima is distributed as the maximum of $( (Y_{d-1,r} + 1)B(Y_{d-1,r} + 1))^{(r)}$ i.i.d. realizations of $W$, from which the lemma follows.
\end{proof}
Recall that $\phi_X$ denotes the probability generating
function for a discrete r.v. $X$.
\begin{lemma}\label{lemma:PROPPELWATER}
Suppose the nodes  and edges of a tree $T=T(r,r-1,d)$ with root $0$ are equipped with i.i.d.
weights generated from a continuous distribution $F$.
Then
\begin{align*}
\E[W_{0} I({0} \in {\cal IG}(T))] = \E[W\phi^r_{{X_{d-1,r-1}}}(F(W))],
\end{align*}
and for every $j\in N(0)$
\begin{align*}
\E[W_{0j} I( (0,j) \in {\cal MG}(T))] = \E[W\phi_{{Y_{d-1,r-1} + Y_{d,r-1}}}(F(W))],
\end{align*}
where $W$ is distributed according to $F$, and random variables $W,X_{d-1,r-1},Y_{d-1,r-1},Y_{d,r-1}$
are independent.
\end{lemma}
\begin{proof}
We first prove the result for independent sets.  By Proposition~\ref{prop:SUPERGOODENOUGH}
and  the definition of $S({0})$,
$\E[W_{0} I(0 \in {\cal IG}(T))] = E[W_{0}I(W_{0} > \max_{j \in C({0})} S(j))]$.
By Lemma~\ref{lemma:WBONUSDIST}, for each $j \in C(0)$, \\$S(j)
\stackrel{D}{=} W^{<X_{d-1,r-1}>}$.  It then follows that $\max_{j \in C(0)} S(j) \stackrel{D}{=}
W^{<X_{d-1,r-1}^{(r)}>}$, which is independent from $W_{0}$. Thus  we
have
\begin{eqnarray*}
\E[W_0 I(0 \in {\cal IG}(T))] &=&
\E[ \E[W_0 I(0 \in {\cal IG}(T)) | W_0]] \\&=&
\E[W_0 \E[I (W^{<X_{d-1,r-1}^{(r)}>} \leq W_0) | W_0]] \\&=&
\E[W_0 \sum_{k=0}^{\infty} (F(W_0))^k \pr(X_{d-1,r-1}^{(r)} = k)] \\&=&
\E[W_0 \phi_{X_{d-1,r-1}^{(r)}}(F(W_0))] \\&=&
\E[W_0 \phi^r_{X_{d-1,r-1}}(F(W_0))].
\end{eqnarray*}
We now prove the result for matchings. From the third part of Proposition~\ref{prop:SUPERGOODENOUGH}
we have
\begin{align*}
\E[W_{0j} I((0,j) \in {\cal IG}(T))] = \E[ W_{0j} I (W_{0j}> \max( MS_{H}(0),  MS_{T_j}(j)))],
\end{align*}
where $H$ is the subgraph of $T$ obtained by deleting $(0,j)$ and $T_j$ - the subtree of $T$ rooted at $j$.
Observe that $H$ is an $r-1$ regular tree with depth $d$, namely it is $T(r-1,d)$,
and $T_j$ is an $r-1$ regular tree with depth $d-1$.
Thus applying Lemma~\ref{lemma:WBONUSDIST}, $MS_H(0)\stackrel{D}{=} W^{<Y_{d,r-1}>}$
and $MS_{{T_j}}(j) \stackrel{D}{=}W^{<Y_{d-1,r-1}>}$. Repeating the line of argument used for independent
sets,  replacing $X_{d-1,r-1}^{(r)}$ with
$Y_{d-1,r-1} + Y_{d,r-1}$, we obtain the result.
\end{proof}

\subsection{Limiting distribution of $X_{d,r}$ and $Y_{d,r}$}\label{section:convsec}
In this subsection we show that the sequences $\lbrace X_{d,r} \rbrace, d\ge 0,$ and $\lbrace Y_{d,r} \rbrace, d\ge 0$  converge
in distribution to  some limiting random variables, by
exploiting their recursive definitions. We then use this convergence along
with Lemma~\ref{lemma:PROPPELWATER} to express the quantities of interest in terms
of the p.g.f. of these limiting random variables.
\begin{lemma}\label{XYCONVERGE}
There exist r.v. $X_{\infty,r}, Y_{\infty,r}$ such that for all $k \geq 0$,
$\lim_{d \rightarrow\infty} \pr(X_{d,r} = k)=\pr(X_{\infty,r} = k)$
and $\lim_{d \rightarrow\infty} \pr(Y_{d,r} = k)=\pr(Y_{\infty,r} = k)$.
\ignore{
Moreover, for every $m\ge 1$,
\begin{align}
\lim_{d \rightarrow\infty}\E[X_{d,r-1}^m]=\E[X_{\infty,r-1}^m],~~
\lim_{d \rightarrow\infty}\E[Y_{d,r-1}^m]=\E[Y_{\infty,r-1}^m], \label{eq:limitsmoments}
\end{align}
and for every $s\in [0,1]$,
\begin{align*}
\lim_{d \rightarrow\infty}\phi_{X_{d,r-1}}(s)=\phi_{X_{\infty,r-1}}(s),~~
\lim_{d \rightarrow\infty}\phi_{Y_{d,r-1}}(s)=\phi_{Y_{\infty,r-1}}(s).
\end{align*}}
\end{lemma}

\begin{proof}
We begin by establishing  the existence of the limit $\lim_{d \rightarrow\infty} \pr(X_{d,r} = k)$
for $k=0$. The case of $k\ge 1$ will be established by induction.
Consider $T=T(r,d)$ with root $0$ whose nodes are weighted i.i.d. with an arbitrary continuous
distribution $F$. From Proposition~\ref{prop:SUPERGOODENOUGH}, part 1, we have that
$S(0)=0$ iff $0\notin {\cal IG}(T)$. Therefore by Lemma~\ref{lemma:WBONUSDIST}
\begin{align*}
\pr(S(0)=0)=\pr(0\notin {\cal IG}(T))=\pr(X_{d,r}=0).
\end{align*}
But the last quantity has a limit as $d\rightarrow\infty$ as asserted by the last part of
Theorem~\ref{theorem:GraphTreeRelation}.

Assume now that the limits exist for all $k'\le k-1$. We have
\begin{align*}
\pr(X_{d,r}=k)={1\over k}\pr(X_{d-1,r}^{(r)}=k-1)={1\over k}\sum_{(k_1,k_2,...,k_r)}\prod_{1\le i\le r}
\pr(X_{d-1,r}=k_i),
\end{align*}
where the sum is over all partitions $(k_1,k_2,...,k_r)$ with $k_i\ge 0, \sum_{1\le i\le r}k_i=k-1$.
Since $k_i\le k-1$ for each $i$, by the inductive assumption
the limits $\lim_{d\rightarrow\infty}\pr(X_{d-1,r}=k_i)$ exist. The same assertion then follows for
$\pr(X_{d,r}=k)$ and the proof is complete.

Define $X_{\infty,r}$ by $\pr(X_{\infty,r}=k)=\lim_{d\rightarrow\infty}\pr(X_{d,r}=k)$.
We need to show that $\sum_k\pr(X_{\infty,r}=k)=1$. Fix $\epsilon>0$ and $K>1/\epsilon$.
Applying Markov's inequality to (\ref{eq:expectationXY}) we have
$1\ge \sum_{0\le k\le K}\pr(X_{d,r}=k)\ge 1-1/K>1-\epsilon$. Then the same applies to the limits
as $d\rightarrow\infty$. The assertion then follows.

The proof for the matching case is similar.
\end{proof}

The recursion properties (\ref{eq:recursionX}) which are used to define $\lbrace X_{d,r} \rbrace , \lbrace Y_{d,r} \rbrace$ carry on to
$X_{\infty,r},Y_{\infty,r}$, which, as a result, satisfy recursive distributional equations.

\begin{lemma}\label{lemma:recursionXYinfty}
The following equality in distribution takes place
\begin{align*}
&X_{\infty,r}\stackrel{D}{=}(X_{\infty,r}^{(r)}+1)B(X_{\infty,r}^{(r)}+1), \\
&Y_{\infty,r}\stackrel{D}{=}\Big((Y_{\infty,r}+1)B(Y_{\infty,r}+1)\Big)^{(r)}.
\end{align*}
\end{lemma}

\begin{proof}
Applying Lemma~\ref{XYCONVERGE}, for each $k>0$,
\begin{align*}
\pr(X_{\infty,r}=k)&=\lim_{d\rightarrow\infty}\pr(X_{d,r}=k) \\
&=\lim_{d\rightarrow\infty}
{1\over k}\sum_{(k_1,...,k_{r})}\prod_{1\le l\le r}\pr(X_{d-1,r}=k_l) \\
&={1\over k}\sum_{(k_1,...,k_{r})}\prod_{1\le l\le r}\pr(X_{\infty,r}=k_l),
\end{align*}
where the sums are over all partitions $(k_1,...,k_{r}), k_l \ge 0, \sum_{1\le l\le r}k_l=k-1$.
But the last expression is exactly the probability that $(X_{\infty,r}^{(r)}+1)B(X_{\infty,r}^{(r)}+1)$
takes value $k$. The assertion then follows. A similar argument shows the identity for $Y_{\infty,r}$.
\end{proof}

\subsection{Solving for the distribution of  $X_{d,r}$ and $Y_{d,r}$}\label{section:genfunsec}
We now show that $\phi_{X_{\infty,r}}(s)$ and $\phi_{Y_{\infty,r}}(s)$ have a very simple
explicit form. We first show that they satisfy  simple differential equations.

\begin{lemma}\label{lemma:derivimport1}
For every $s\in [0,1)$
\begin{align*}
&\frac{d}{ds}\phi_{X_{\infty,r}}(s)  =\phi^{r}_{X_{\infty,r}}(s), \\
&\frac{d}{ds}
\phi^{\frac{1}{r}}_{Y_{\infty,r}}(s) = \phi_{Y_{\infty,r}}(s).
\end{align*}
\end{lemma}
\begin{proof}
We first prove the identity for $X_{\infty,r}$.  Applying Lemma~\ref{lemma:recursionXYinfty},
\begin{align*}
\phi_{X_{\infty,r}}(s) =  \pr(X_{\infty,r} = 0) +
\sum_{k=0}^{\infty} \frac{1}{k+1}s^{k+1} \pr(X_{\infty,r}^{(r)} = k).
\end{align*}
Thus since the p.g.f. of any non-negative integer-valued r.v. is differentiable on [0,1), and  can be
 differentiated term-by-term, we obtain
\begin{align*}
\frac{d}{ds} \phi_{X_{\infty,r}}(s)
& = \frac{d}{ds} \sum_{k=0}^{\infty} \frac{1}{k+1}s^{k+1} \pr (X_{\infty,r}^{(r)} = k)\\
&= \sum_{k=0}^{\infty} s^{k} \pr (X_{\infty,r}^{(r)} = k) \\
&=\phi^{r}_{X_{\infty,r}}(s).
\end{align*}
As for for $Y_{\infty,r}$ we have from Lemma~\ref{lemma:recursionXYinfty} that
$\phi^{1\over r}_{Y_{\infty,r}}(s)$ is equal to the p.g.f. of $(Y_{\infty,r}+1)B(Y_{\infty,r}+1)$.
Therefore
\begin{align*}
\frac{d}{ds} \phi^{1\over r}_{Y_{\infty,r}}(s)&=
\sum_{k=0}^{\infty} \frac{d}{ds}\frac{1}{k+1}s^{k+1} \pr(Y_{\infty,r} = k) \\
&=\phi_{Y_{\infty,r}}(s).
\end{align*}
\end{proof}

We now  solve for the p.g.f. of $X_{\infty,r},Y_{\infty,r}$.
\begin{proposition}\label{prop:formofpgf}
For every $s \in [0,1]$,
\begin{align*}
&\phi_{X_{\infty,r}}(s) = (r -(r-1)s)^{-\frac{1}{r-1}}\\
&\phi_{Y_{\infty,r}}(s) =  (r- (r-1)s)^{-\frac{r}{r-1}}.
\end{align*}
\end{proposition}
\begin{proof}
Applying the chain rule
and Lemma~\ref{lemma:derivimport1}
\begin{align*}
\frac{d}{ds}
\phi^{-(r-1)}_{_{\infty,r}}(s) &=
-(r-1)(\phi_{X_{\infty,r}}(s))^{-r}\frac{d}{ds}\phi_{X_{\infty,r}}(s)\\
&=-(r-1)(\phi_{X_{\infty,r}}(s))^{-r}\phi^{r}_{X_{\infty,r}}(s) \\
&=-(r-1).
\end{align*}
We conclude that $\phi_{X_{\infty,r}}(s)=(c-(r-1)s)^{-{1\over r-1}}$ for some constant $c$ for all $s\in [0,1)$.
Since $\lim_{s\uparrow 1}\phi_{X_{\infty,r}}(s)=\phi_{X_{\infty,r}}(1)=1$, we conclude
that $c=r$ and therefore
$\phi_{X_{\infty,r}}(s)=(r- (r-1)s)^{-\frac{1}{r-1}}$ and the required identity is established.

Similarly, we find
\begin{align*}
\frac{d}{ds}
(\phi_{Y_{\infty,r}}(s))^{-{r-1\over r}} &=
\frac{d}{ds}
\Big(\phi^{{1\over r}}_{Y_{\infty,r}}(s)\Big)^{-(r-1)}\\
&=-(r-1)\Big(\phi^{{1\over r}}_{Y_{\infty,r}}(s)\Big)^{-r}
\frac{d}{ds}\phi^{{1\over r}}_{Y_{\infty,r}}(s) \\
&=-(r-1)(\phi_{Y_{\infty,r}}(s))^{-1}\phi_{Y_{\infty,r}}(s) \\
&=-(r-1).
\end{align*}
Using this and $\phi_{Y_{\infty,r}}(1)=1$ the required identity is established.
\end{proof}

\subsection{Proofs of Theorems ~\ref{theorem:MWISWEIGHT} and \ref{theorem:MWMWEIGHT}}
We now have all the necessary results to complete the proofs of our main theorems.

\begin{proof}[Proof of Theorem~\ref{theorem:MWISWEIGHT}]
Applying the last part of Lemma~\ref{XYCONVERGE} and the Dominated Convergence Theorem (see \cite{durrett})
we have that if $W\stackrel{D}{=}F$, then
\begin{align*}
\lim_{d\rightarrow\infty}\E[W\ \phi^r_{{X_{d-1,r-1}}}(F(W))]&=\E[W\phi^r_{{X_{\infty,r-1}}}(F(W))].
\end{align*}
Here $W$ serves as a dominating random variable.
Applying Proposition~\ref{prop:formofpgf} the right-hand side of this expression equals
\begin{align*}
\E[W(r-1 -(r-2)F(W))^{-\frac{r}{r-2}}]&=\int_{0}^{\infty} x \Big((r-1) - (r-2)F(x)\Big)^{-\frac{r}{r-2}}f(x)dx.
\end{align*}
Now observe that $\E[W[{\cal IG}]]=\sum_{i\in V(G)}\E[W_iI(i\in{\cal IG})]$.
Applying part (\ref{eq:WgraphWtreeIS}) of Theorem~\ref{theorem:GraphTreeRelation},
Lemma~\ref{lemma:PROPPELWATER}, and letting $d\rightarrow\infty$ we obtain the result.
\end{proof}

\begin{proof}[Proof of Theorem~\ref{theorem:MWMWEIGHT}]
Observe that $\E[W[{\cal MG}]]=\sum_{e\in E(G)}\E[W_eI(e\in{\cal MG})]$. The rest of the proof is similar to the
case for independent sets.
\end{proof}

\section{The variance of $GREEDY$}\label{section:varrsec}
In this section we prove our main results on the variance of $GREEDY$.

\begin{proof}[Proof of Theorem~\ref{theorem:VAR12}]
Since $W[{\cal IG}]=\sum_{i\in V}W_iI(i\in {\cal IG})$, we have
\begin{align}
Var(W[{\cal IG}])&=\sum_{i,j\in V}\Big(\E[W_iW_jI(i,j\in {\cal IG})]-
\E[W_iI(i\in {\cal IG})]\E[W_jI(j\in {\cal IG})]\Big) \notag\\
&=\sum_{i\in V}\Big(\E[W_i^2I(i\in {\cal IG})]-(\E[W_iI(i\in {\cal IG})])^2\Big) \notag\\
&+\sum_{i\in V}\sum_{d\ge 0}\sum_{j\in N_{d+1}(i)\setminus N_d(i)}
\Big(\E[W_iW_jI(i,j\in {\cal IG})]-\E[W_iI(i\in {\cal IG})]\E[W_jI(j\in {\cal IG})]\Big) \notag\\
&\le n\E[W^2]\notag\\
&+\sum_{i\in V}\sum_{d\ge 0}\sum_{j\in N_{d+1}(i)\setminus N_d(i)}
\Big(\E[W_iW_jI(i,j\in {\cal IG})]-\E[W_iI(i\in {\cal IG})]\E[W_jI(j\in {\cal IG})]\Big). \label{eq:varexpanded}
\end{align}
Our proof approach is to show that the terms in parenthesis are sufficiently close to each other,
provided that the distance
between nodes $i$ and $j$ is sufficiently large.

Fix an arbitrary $i\in V$ and $j\in N_{d+1}(i)\setminus N_d(i)$ for $d\ge 2$. Recall the notion of the influence blocking
subgraph from Section~\ref{section:BLOCKINGSUBTREES}. Denote $IB(N(i))$ and $IB(N(j))$ by $H_i$ and $H_j$ for short.
Let $l=\lfloor d/2 \rfloor - 1$. Consider the event
${\cal E}\triangleq \big(H_i\subset N_l(i)\wedge H_j\subset N_l(j)\big)^c$.

We have
\begin{align}
\E[W_iW_jI(i,j\in {\cal IG})]&=\E[W_iW_jI(i,j\in {\cal IG},H_i\subset N_l(i),H_j\subset N_l(j))] \notag\\
&+\E[W_iW_jI(i,j\in {\cal IG})I({\cal E})]. \label{eq:overlap}
\end{align}
We first analyze the second summand.
\begin{align*}
\E[W_iW_jI(i,j\in {\cal IG})I({\cal E})]
&\le \E[W_iW_jI({\cal E})]\\
&\le \E[W_iW_j(1-I(H_i\subset N_l(i)))]+\E[W_iW_j(1-I(H_j\subset N_l(j)))]\\
&=\E[W_i(1-I(H_i\subset N_l(i)))]\E[W_j]+\E[W_j(1-I(H_j\subset N_l(j)))]\E[W_i]
\end{align*}
where the equality holds since both $W_i$ and the event $H_i\subset N_l(i)$ depend only on the weight
configuration inside $N_{l+1}(i)$ which does not contain node $j$ and vice verse. Next, applying the second part of
Lemma~\ref{lemma:ind11} and Lemma~\ref{lemma:shortandsweet} we have
\begin{align}
\E[W_i(1-I(H_i\subset N_l(i)))]&=\E[W_i](1-\pr(H_i\subset N_l(i))) \notag\\
&\le \E[W_i]r(r-1)^l/(l+1)!. \label{eq:notinNi}
\end{align}
Thus we obtain
\begin{align}
\E[W_iW_jI(i,j\in {\cal IG})I({\cal E})]&\le 2\E[W_i]\E[W_j]r(r-1)^l/(l+1)! \notag\\
&=2\E[W]^2r(r-1)^l/(l+1)! \label{eq:CalE}
\end{align}
We now analyze the first summand in (\ref{eq:overlap}). Let $\hat H_i=H_i\cap N_l(i),\hat H_j=H_j\cap N_l(j)$.
Namely, $\hat H_i$ and $\hat H_j$ are the subgraphs of $N_l(i)$ and $N_l(j)$ induced by nodes
$V(H_i)\cap V(N_l(i))$ and  $V(H_j)\cap V(N_l(j))$, respectively.
Observe that the random variables $W_iI(i\in {\cal IG}(\hat H_i), \hat H_i=H_i)$ and
$W_jI(j\in {\cal IG}(\hat H_j), \hat H_j=H_j)$ are independent. Indeed, since $I(\hat H_i=H_i) = I(H_i \in N_l(i))$ and $I(\hat H_j=H_j) = I(H_j \in N_l(j))$,
they are completely determined
by the weights inside $N_{l+1}(i)$ and $N_{l+1}(j)$ (respectively) and those do not intersect. Therefore
\begin{align}
\E[W_iW_j&I(i\in {\cal IG}(\hat H_i), \hat H_i=H_i,j\in {\cal IG}(\hat H_j), \hat H_j=H_j)] \notag\\
&=
\E[W_iI(i\in {\cal IG}(\hat H_i), \hat H_i=H_i)]
\E[W_jI(j\in {\cal IG}(\hat H_j), \hat H_j=H_j)] \label{eq:independence}
\end{align}
On the other hand
\begin{align*}
I(i\in {\cal IG}(\hat H_i), \hat H_i=H_i)&=I(i\in {\cal IG}(\hat H_i), H_i\subset N_l(i))\\
&=I(i\in {\cal IG}, H_i\subset N_l(i))
\end{align*}
where the second equality follow from Lemma~\ref{lemma:SBOUNDEDSAYSITALL}. Similarly we obtain
\begin{align*}
I(j\in {\cal IG}(\hat H_j), \hat H_j=H_j)&=I(j\in {\cal IG}, H_j\subset N_l(j)).\\
\end{align*}
Thus, we can rewrite (\ref{eq:independence}) as
\begin{align*}
\E[W_iW_j&I(i,j\in {\cal IG},H_i\subset N_l(i),H_j\subset N_l(j))] \\
&=\E[W_iI(i\in {\cal IG}, H_i\subset N_l(i))]\E[W_jI(j\in {\cal IG}, H_j\subset N_l(j))].
\end{align*}
We recognize the left-hand side of this equation as the first summand in (\ref{eq:overlap}).
Returning to (\ref{eq:overlap}) we obtain
\begin{align}
\Big|\E[W_iW_jI(i,j\in {\cal IG})]&-
\E[W_iI(i\in {\cal IG}, H_i\subset N_l(i))]\E[W_jI(j\in {\cal IG}, H_j\subset N_l(j))]\Big| \notag\\
\le
&2\E[W]^2r(r-1)^l/(l+1)! . \label{eq:bound1}
\end{align}
Also we have
\begin{align*}
\E[W_iI(i\in {\cal IG})]&=
\E[W_iI(i\in {\cal IG},H_i\subset N_l(i))]+\E[W_iI(i\in {\cal IG}(H_i))(1-I(H_i\subset N_l(i)))] ,
\end{align*}
and
\begin{align*}
\E[W_iI(i\in {\cal IG}(H_i))(1-I(H_i\subset N_l(i)))]\le \E[W_i(1-I(H_i\subset N_l(i)))]
\le \E[W]r(r-1)^l/(l+1)!,
\end{align*}
where the second inequality is (\ref{eq:notinNi}). It follows
\begin{align*}
\Big|\E[W_iI(i\in {\cal IG})]-\E[W_iI(i\in {\cal IG},H_i\subset N_l(i))]\Big|
\le \min( \E[W]r(r-1)^l/(l+1)!, \E[W_iI(i\in {\cal IG})]).
\end{align*}
A similar inequality holds for $j$. Putting these two bounds together we obtain
\begin{align*}
\Big|\E[W_iI(i\in {\cal IG})]&\E[W_jI(j\in {\cal IG})]-
\E[W_iI(i\in {\cal IG},H_i\subset N_l(i))]
\E[W_jI(j\in {\cal IG},H_j\subset N_l(j))]\Big| \\
&\le 2\E[W]^2r(r-1)^l/(l+1)!,
\end{align*}
where  trivial bounds $\E[W_iI(i\in {\cal IG})]\le \E[W_i],
\E[W_jI(j\in {\cal IG})]\le \E[W_j]$ are used.
Combining with bound (\ref{eq:bound1}) we obtain
\begin{align*}
\Big|\E[W_iW_jI(i,j&\in {\cal IG})]-\E[W_iI(i\in {\cal IG})]\E[W_jI(j\in {\cal IG})]\Big| \\
&\le 4\E[W]^2r(r-1)^l/(l+1)! \\
&= 4 \frac{r}{r-1} \E[W^2] (r-1)^{l+1}/(l+1)! .
\end{align*}
We now use this estimate in (\ref{eq:varexpanded}). Observe that $|N_{d+1}(i)\setminus N_d(i)| \leq r(r-1)^d$.
Recall that $l+1=\lfloor d/2\rfloor $. Then for each $i$, considering the cases of
odd and even $d$ separately and observing that the estimate is also trivially an upper bound for the $d=0,1$ cases, the double sum
$\sum_{d\ge 0}\sum_{j\in N_{d+1}(i)\setminus N_d(i)}$ in (\ref{eq:varexpanded}) is upper bounded by
\begin{align*}
4 \frac{r}{r-1} \E[W^2] \sum_{d \ge 0} r (r-1)^d {(r-1)^{\lfloor \frac{d}{2} \rfloor}\over (\lfloor d/2\rfloor)!}
&= 4 \frac{r^2}{r-1} \E[W^2] \sum_{k\ge 0} {(r-1)^{3k}\over k!}
+4 \frac{r^2}{r-1} \E[W^2] \sum_{k\ge 0} {(r-1)^{3k+1}\over k!} \\
&<8 \frac{r^2}{r-1} \E[W^2] \sum_{k\ge 0} {(r-1)^{3k+1}\over k!} \\
& < 8\E[W^2] r^2 \exp((r-1)^3).
\end{align*}
Our final upper bound on $Var(W[{\cal IG}])$ becomes
\begin{align*}
n\E[W^2]+8n\E[W^2]r^2\exp((r-1)^3)<9n\E[W^2]r^2\exp((r-1)^3)
\end{align*}
This completes the proof.  The proof for matchings follows similarly, and is omitted.
\end{proof}

\section{Numerical results}\label{section:NUMERICS}
In this section we numerically evaluate the performance of $GREEDY$ in several settings, and compare
our results to the prior work.
We first compare our bound, marked NEW in the table below, on the cardinality (normalized by the number of nodes)
of a $MIS$ in an $r$-regular
graph of girth at least $g$ (Corollary~\ref{coro:MICOR}) to the previous bounds in~\cite{S.91} and \cite{L.}.  The bounds of \cite{S.91} are coming from their Theorem 3 (when $g < 127$) and their Theorem 4 (when $g \geq 127)$, with $w_i = 1$ for all $i$ (their formulas involve a notion of weighted girth).  The bounds of \cite{L.} are coming from their Table 2.  Omitted values are those for which no corresponding results are given or the given bounds are trivial.  Certain values of the form $2k + 3$ are emphasized to be compatible with the Table 2 given in \cite{L.}.  All values are rounded up to the nearest thousandth.
\begin{table}[h]
\caption{Comparison of bounds for the cardinality of $MIS$ in $r$-regular large-girth graphs}
\centering
\begin{tabular}{|c|ccc|ccc|ccc|}
\hline
\myfrac{g}{r} &\multicolumn{3}{c|}{5}&\multicolumn{3}{|c|}{7}&\multicolumn{3}{|c|}{10} \\
\hline
\ &  NEW & $\cite{S.91}$ & $\cite{L.}$ &  NEW & $\cite{S.91}$ & $\cite{L.}$ &  NEW & $\cite{S.91}$ & $\cite{L.}$\\
\hline
50 & .302 & .288 & - & .256 & .239 & - & .160 & .194 & -  \\
100 & .302 & .294 & - & .256 & .243 & - & .211 & .197 & -  \\
203 & .302 & .304 & .262 & .256 & .250 & - & .211 & .201 & .169  \\
403 & .302 & .306 & .277 & .256 & .251 & - & .211 & .202 & .184  \\
2003 & .302 & .308 & .294 & .256 & .252 & - & .211 & .203 & .202  \\
\hline
\end{tabular}
\label{tabble}
\end{table}
As we see our new bounds are the strongest for many calculated values of $g$ and $r \geq 7$.  Recall that
for $r \geq 7$, our bounds are asymptotically (as $g \rightarrow \infty$) equivalent to those of \cite{L.},
and superior to those of \cite{S.91}.  Note that our bounds converge to their limit much faster than the bounds
of \cite{S.91} and \cite{L.}.
\newpage
 We now give our bounds for the cardinality of a $MM$ (also normalized by the number of nodes) in an $r$-regular graph of girth
at least $g$ (Corollary~\ref{coro:MMCOR}).
These are the first results for $MM$ in this setting.
\begin{table}[h]
\caption{Bounds for the cardinality of $MM$ in $r$-regular large-girth graphs}
\centering
\begin{tabular}{|c|c|c|c|c|c|c|c|}
\hline
\myfrac{g}{r} &3&4&5&6&7&10&13 \\
\hline
25 &  .437  &  .427  & -  &  -   & -    & - & - \\
40 &  .438  &  .444  & .450  &  .454   & .424    & -& -  \\
50 &  .438 & .444    & .450  & .455 & .459 & - &-\\
75 &  .438 & .444    & .450  & .455 & .459 & .468&.449 \\
100 &  .438 & .444    & .450  & .455 & .459 & .468&.473 \\

\hline
\end{tabular}
\label{tabble2}
\end{table}

Note that as $r$ increases, the asymptotic (in $r$) size of a $MM$ approaches that of a perfect matching ($\frac{n}{2}$),
as expected from Corollary~\ref{coro:MMCORunweight}.

 We now give our results for $MWIS$ and $MWM$ with i.i.d Exp(1) (exponentially distributed with parameter $1$)
 weights, and compare to the results
 given in \cite{GamarnikNowickiSwirscszExpDyn}.  The $GREEDY$ columns show the expected asymptotic weight (normalized by
 the number of nodes)
 of the weighted independent set and matching returned by $GREEDY$
 as given in Theorems~\ref{theorem:MWISWEIGHT} and~\ref{theorem:MWMWEIGHT},
 while the \cite{GamarnikNowickiSwirscszExpDyn} columns reflect the expected asymptotic weight of a true $MWIS$ and $MWM$ as computed in \cite{GamarnikNowickiSwirscszExpDyn}.  We only give results for $r$-regular graph with limiting girth, as no results for fixed girth are given in \cite{GamarnikNowickiSwirscszExpDyn}.

\vspace{+.1in}
\begin{table}[h]
\caption{Exact MWIS and MWM vs. GREEDY for $r$-regular large-girth graphs with i.i.d Exp(1) weights}
\centering
\begin{tabular}{|c|cc|cc|}
\hline
\ &\multicolumn{2}{c|}{$MWIS$}&\multicolumn{2}{|c|}{$MWM$} \\
\hline
$r$  &  $\cite{GamarnikNowickiSwirscszExpDyn}$&$GREEDY$&$\cite{GamarnikNowickiSwirscszExpDyn}$&$GREEDY$ \\
\hline
3   & .6077 & .5966 & .7980 & .7841\\
4   & $.5631^{1}$ & .5493 & .9022 & .8826\\
5   & - & .5119 & .9886 & .9643\\
10   & - & .3967 & 1.282 & 1.242 \\
\hline
\end{tabular}
\label{tabble3}
\end{table}

In all cases, $GREEDY$ is nearly optimal.
\footnotetext[1]{reported incorrectly in \cite{GamarnikNowickiSwirscszExpDyn}.}

\section{Conclusion}\label{section:CONCLUSION}
We have provided new results for the performance of a simple
randomized greedy algorithm for finding large independent sets and
matchings in regular graphs with large finite girth. This provided
new constructive and existential results in several settings. One of
the interesting insights from this work is demonstrating a
correlation decay property of the greedy algorithm which  aids
greatly the analysis of this algorithm. In addition to several
explicit bounds on the sizes and weights of independent sets and
matchings produced by the $GREEDY$ algorithm, we established
concentration results by bounding the variance of the values
produced by $GREEDY$. As a byproduct, this shows that $GREEDY$ is
very robust w.r.t. random choices: running it several times will
produce roughly the same result in terms of the cardinality of the
produced independent set and matching.

\section*{Acknowledgements}\label{section:ACK}
The authors would like to thank Alan Frieze and Nick Wormald for a
very helpful correspondence explaining the state of the art results
in the area.  The authors also thank Dimitriy Katz for his insights
into the performance of the $GREEDY$ algorithm, and Theo Weber for
several interesting  discussions about the decay of correlations
phenomenon.

\bibliographystyle{amsalpha}

\providecommand{\bysame}{\leavevmode\hbox to3em{\hrulefill}\thinspace}
\providecommand{\MR}{\relax\ifhmode\unskip\space\fi MR }
\providecommand{\MRhref}[2]{%
  \href{http://www.ams.org/mathscinet-getitem?mr=#1}{#2}
}
\providecommand{\href}[2]{#2}

\end{document}